\RequirePackage{xcolor}
\documentclass[journal,twoside,web]{ieeecolor}  
\usepackage{etoolbox}
\makeatletter
\@ifundefined{color@begingroup}%
  {\let\color@begingroup\relax
   \let\color@endgroup\relax}{}%
\def\fix@ieeecolor@hbox#1{%
  \hbox{\color@begingroup#1\color@endgroup}}
\patchcmd\@makecaption{\hbox}{\fix@ieeecolor@hbox}{}{\FAILED}
\patchcmd\@makecaption{\hbox}{\fix@ieeecolor@hbox}{}{\FAILED}
\usepackage{generic}
\usepackage{cite}
\definecolor{amethyst}{rgb}{0.6, 0.4, 0.8}
\usepackage{amsmath,amssymb,amsfonts, bm, nicefrac}
\usepackage{graphicx}
\usepackage{comment}
\usepackage{algorithm}
\usepackage[noend]{algpseudocode}
\usepackage{fp, tikz}
\usepackage{pgfplots}
\usepgfplotslibrary{fillbetween}
\pgfplotsset{compat=newest}
\usepackage{soul}
\usepgfplotslibrary{units}
\usetikzlibrary{spy,backgrounds}
\usepackage{pgfplotstable}
\usepackage{booktabs} 
\usepackage{tabularx}
\usepackage{epstopdf}
\usetikzlibrary{arrows,shapes,backgrounds,patterns,fadings,matrix,arrows,calc,
	intersections,decorations.markings,
	positioning,external,arrows.meta}
\tikzset{
	block/.style = {draw, rectangle,
		minimum height=1cm,
		minimum width=2cm},
	input/.style = {coordinate,node distance=1cm},
	output/.style = {coordinate,node distance=6cm},
	arrow/.style={draw, -latex,node distance=2cm},
	pinstyle/.style = {pin edge={latex-, black,node distance=2cm}},
	sum/.style = {draw, circle, node distance=1cm},
}
\usetikzlibrary{ext.paths.ortho}
\usepgfplotslibrary{fillbetween}
\makeatletter
\let\NAT@parse\undefined
\makeatother
\usepackage{hyperref}
\usepackage{url}            
\usepackage{textcomp}
\usepackage[capitalize]{cleveref}
\def\BibTeX{{\rm B\kern-.05em{\sc i\kern-.025em b}\kern-.08em
    T\kern-.1667em\lower.7ex\hbox{E}\kern-.125emX}}
\newtheorem{definition}{Definition}
\newtheorem{assumption}{Assumption}

\newtheorem{theorem}{Theorem}
\newtheorem{lemma}{Lemma}

\newtheorem{proposition}{Proposition}

\newcommand{\real}{\mathbb{R}}

\newcommand{\revision}[1]{\textcolor{black}{#1}}
\begin{document}
\title{\revision{Safe Event-triggered Gaussian Process Learning for  Barrier-Constrained Control}}
\author{Armin Lederer$^{*}$,~\IEEEmembership{Member,~IEEE}, Azra Begzadi\'c$^{*}$,~\IEEEmembership{Graduate Student Member,~IEEE,} Sandra Hirche,~\IEEEmembership{Fellow,~IEEE,} Jorge Cort{\'e}s,~\IEEEmembership{Fellow,~IEEE,} Sylvia Herbert,~\IEEEmembership{Member,~IEEE}
\thanks{\textit{${}^*$ Armin Lederer and Azra Begzadi\'c contributed equally to this work.}}
\thanks{A. Lederer is with the Department of Electrical and Computer Engineering, College of Design and Engineering, National University of Singapore (e-mail: armin.lederer@nus.edu.sg).}
\thanks{A. Begzadi\'c, J. Cort\'es and S. Herbert are with the Department of Mechanical and Aerospace Engineering, University of California, San Diego (e-mail: abegzadic@ucsd.edu, cortes@ucsd.edu, sherbert@ucsd.edu).}
\thanks{S. Hirche is with the Chair of Information-oriented Control, TUM School of Computation, Information and Technology, Technical University of Munich, Germany (e-mail: hirche@tum.de).}
}

\maketitle

\begin{abstract} 
While control barrier functions (CBFs) are employed
in addressing safety, control synthesis methods based on them generally rely on accurate system dynamics. This is a critical limitation, since the dynamics of complex systems are often not fully known. Supervised machine learning techniques hold great promise for alleviating this weakness by inferring models from data. We propose a novel 
\revision{approach for safe event-triggered learning of Gaussian process models in CBF-based continuous-time control for unknown control-affine systems. By applying a finite excitation at triggering times, our approach ensures a sufficient information gain to maintain the feasibility of the CBF-based safety condition with high probability. Our approach probabilistically guarantees safety based on a suitable GP prior and rules out} Zeno behavior in the triggering scheme. 
The effectiveness of the proposed approach and theory is demonstrated in simulations.
\end{abstract}
\begin{IEEEkeywords}
Event-triggered learning, safety-critical control, Gaussian processes, control barrier functions.

\end{IEEEkeywords}
\definecolor{azrablue}{RGB}{0, 103, 180}
\definecolor{gr}{RGB}{85, 139, 47}
\definecolor{ared}{RGB}{184, 15, 10}
\definecolor{asiva}{RGB}{94, 94, 94}
\section{Introduction}

\label{sec:introduction}
\IEEEPARstart{C}{ontrol} barrier functions (CBFs) are commonly used to guarantee the safety of nonlinear systems \cite{Ames2017}. Given a valid CBF for a control-affine system, safe control inputs can be efficiently synthesized online using a quadratic program (QP). This QP formulation generally assumes perfect knowledge of system
dynamics, which can be challenging to derive analytically for many applications. 
For example, real-world systems such as autonomous vehicles, industrial machinery, and medical robots have inherent complexity, which often prevents accurate system identification using classical techniques.\looseness=-1

Recently, Gaussian processes (GPs) have gained attention for \revision{learning models} in control due to their strong theoretical foundations \cite{Rasmussen2006}. \revision{In particular, the measure of model uncertainty that GPs provide along with predictions is beneficial for safety-critical problems, as it allows leveraging the robustness of control algorithms. This idea is commonly exploited by combining GPs with CBFs~\cite{wang_safe_2018,Jagtap2020ControlBF},} which results in a safety filter that can be formulated as second-order cone program (SOCP). 
\revision{Note that we focus on continuous-time systems in this article as CBF methods for the discrete-time setting exhibit a fundamentally different structure \cite{cheng2019end}.
}\looseness=-1

\subsubsection*{Related Work}
\revision{Early works on CBF-based continuous-time control with GP models have considered uncertain models, in which} the effect of the control input on the dynamics is known, \revision{i.e., there is only uncertainty about the drift term. Hence,} GP accuracy guarantees can be directly employed for adapting the robustness of CBF conditions to the model error \cite{wang_safe_2018, Jagtap2020ControlBF}. \revision{The resulting optimization problems defining these control laws remain QPs, such that feasibility is guaranteed despite model inaccuracies.
Since inaccurate GP models negatively affect their conservatism, online learning can be used to improve their performance.%
} 
In particular, event-triggered learning~\cite{Solowjow2018} is a promising paradigm wherein data is collected and the model is updated only when certain conditions are satisfied, e.g., model uncertainty exceeds a prescribed threshold.
Event-triggered learning applied to GP model inference is exploited for Lyapunov-based stabilization \cite{Umlauft2019FeedbackLB, zhang_event-triggered_2024}, for which exceptional data efficiency is demonstrated. 
Due to the close relationship between Lyapunov functions and barrier functions~\cite{Ames2017}, it is possible to extend these concepts to CBFs~\revision{\cite{wu_safe_2022, zhang_gaussian_2025}.}

By \revision{encoding} control-affine model structures in the kernel used in GPs, the \revision{adaptation of robustness in CBFs} can be straightforwardly extended to obtain safe control inputs despite \revision{fully unknown control-affine systems \cite{Long2022SafeCS, zhang_real-time_2024}.} Since the resulting optimization problems for determining control inputs become SOCPs rather than QPs, safety can still be efficiently ensured in principle. However, this optimization problem is known to lack feasibility guarantees in general, see e.g.~\cite{Mestres2023FeasibilityAR}, \revision{such that safety guarantees rely on the accuracy of the GP model \cite{castaneda2023}.} \revision{Worst-case back-up strategies relying on prior knowledge can partially remedy this issue \cite{Lederer2023CBF}, but they generally deteriorate performance. While online learning strategies seem appealing to mitigate this issue, their design is not trivial as the chosen control input has a crucial impact on learning \cite{Lederer2021HowTD} in contrast to the situation where only the drift term is unknown. Based on a feasibility analysis of the CBF-based safety conditions for GP models, safe directions for event-triggered learning are derived in \cite{castaneda2023}. However, the necessary magnitude of control inputs in this direction is potentially unbounded.
Moreover, the actual benefit of added data points can be arbitrarily small, such that the occurrence of multiple update events at the same time cannot be excluded. Hence,} Zeno behavior, i.e., an infinite number of trigger events in a finite time interval \cite{Goebel2012HybridDS}, is possible. \revision{Overall, the analysis in \cite{castaneda2023} examines the effect of uncertainty on the CBF 
safety guarantees and, to maintain recursive feasibility, imposes extra assumptions on the GP model at every update.
%
Thus, crucial problems remain open in the design of a practically implementable event-triggered learning approach.}\looseness=-1
%
\subsubsection*{Contribution}
\revision{We propose a novel approach for event-triggered learning of GP models that enables the application of CBF-based continuous-time controllers to fully unknown control-affine dynamics. The key contributions for the development of our implementable approach are the following:
}
\begin{itemize}
    \item We derive a closed-form expression for the required control inputs at GP model updates. \revision{Thereby, we ensure that finite controls are sufficient to maintain the feasibility of CBF conditions with event-triggered learning;}
    \item \revision{We prove the recursive feasibility of our safe event-triggered learning for fully unknown control-affine dynamics. By analyzing the complete model update steps, no assumptions on the GP beyond the prior are required;}
    \item \revision{
    We show that our approach allows us to tune the minimal inter-event times based on the magnitude of the control inputs applied as excitation signal. Since this property excludes Zeno behavior, we ensure the well-posedness of event-triggered learning -- a crucial property in practice.
    }
\end{itemize}
\subsubsection*{Outline} The remainder of this article is structured as follows: Section \ref{sec:problemstatement} defines the problem formulation. In Section \ref{sec:con-time}, we present a novel event-triggered learning approach with a switching control law that guarantees safety and feasibility. 
Finally, in Section \ref{sec:numerical_eval}, we illustrate the proposed method in a numerical simulation of an adaptive cruise control example, before we conclude the paper in Section \ref{sec:conclusion}.

\section{Problem Formulation and Preliminaries}
\label{sec:problemstatement}
\subsection{Notation} 
Vectors and matrices are denoted by bold lower and upper-case symbols, respectively.
We denote by $\real_{>0}$ and $\real_{\ge 0}$
the sets of real positive numbers without and with zero, respectively, and by $\mathbb{N}$ the set of natural numbers. The Euclidean norm is denoted by $\left \| \cdot \right \|$. A function $f: \mathbb{R}^n \rightarrow \mathbb{R}$ is locally Lipschitz if, for every compact set $\mathcal{S} \subset \mathbb{R}^n$, there exists $L>0$ such that $\|f(x)-f(y)\| \leq L\|x-y\|$, for all $x, y \in \mathcal{S}$. $\nabla_{\bm{x}} f$ denotes the gradient of a function $f$ with respect to $\bm{x}$. A continuous function $\alpha: \real_{>0} \rightarrow \real_{>0}$ is of extended class $\mathcal{K}$ function if it is strictly increasing, $\alpha(0)=0$, $\lim _{r \rightarrow \infty} \alpha(r)=\infty$ and $\lim _{r \rightarrow-\infty} \alpha(r)=-\infty$. The Gaussian distribution with mean $\mu \in \mathbb{R}$ and variance $\sigma^2 \in \real_{>0}$ is denoted by $\mathcal{N}\left(\mu, \sigma^2\right)$. The function $\operatorname{diag}([x_1, x_2, \ldots x_n])$ constructs a diagonal matrix where the elements $(x_1, x_2, \ldots x_n)$ are scalar values representing the diagonal entries.
$\bm{I}_{N}$ denotes the $N \times N$ identity matrix. 
The absolute value of a scalar value $x$ is denoted by $|x|$. The sign function, denoted as $\mathrm{sgn}(\cdot)$, returns $-1$ for negative values, $0$ for zero, and $1$ for positive values.

\subsection{Problem Setting}

We consider \revision{nonlinear control-affine} dynamics of the form 
\begin{align}\label{eq:si_sys}
    \revision{\dot{\bm{x}}=\bm{f}(\bm{x})+\bm{G}(\bm{x})\bm{u}}
\end{align}
where $\bm{x}\in\mathcal{X}\subseteq\mathbb{R}^{n}$, \revision{$\bm{u}\in\mathcal{U}\subseteq\mathbb{R}^{m}$,} and $\bm{f}:\mathbb{R}^{n}\rightarrow\revision{\mathbb{R}^{n}}$ and $\bm{G}:\mathbb{R}^{n}\rightarrow\revision{\mathbb{R}^{n}\times\mathbb{R}^{m}}$ are unknown, Lipschitz functions. Throughout the paper, the argument
$t$ in $\bm{x}(t)$ and $\bm{u}(t)$ is omitted for brevity whenever possible.

Based on this system description, we consider the problem of designing a control law $\bm{\pi}:\mathcal{X} \rightarrow \mathbb{R}$ which ensures the safety of the system \eqref{eq:si_sys}. The primary goal of safety is to constrain all system trajectories to a predefined safe set $\mathcal{C}$, which we assume to exhibit no isolated points. 
We define this set as the zero-super level set $\mathcal{C} =\left\{\bm{x} \in \mathcal{X}: \psi(\bm{x}) \geq 0\right\}$
of a continuously differentiable function $\psi: \mathcal{X} \rightarrow \mathbb{R}$. Therefore, safety essentially reduces to forward invariance of~$\mathcal{C}$, as formalized in the following definition.
\begin{definition}[Safety \revision{\cite{cohen_safety-critical_2024}}] \label{def:safety}
A system \eqref{eq:si_sys} is safe with respect to the set $\mathcal{C}$ if this set is forward control invariant, i.e., for some $u \in \mathcal{U}$ starting at any initial condition $\bm{x}_0 \in \mathcal{C}$, it holds that $\bm{x}(t) \in \mathcal{C}$  for $\bm{x}(0)=\bm{x}_0 $ and all $t\geq 0$.
\end{definition}
A common method to show this form of safety relies on the concept of barrier function (CBF), a powerful tool to certify the safety of a wide range of control laws.
\begin{lemma}[Control Barrier Functions \revision{\cite{cohen_safety-critical_2024}}]\label{lem:cbf}
    Consider a dynamical system \eqref{eq:si_sys} and a set $\mathcal{C}$ defined by a continuously differentiable function $\psi: \mathcal{X} \rightarrow \mathbb{R}$. If there exists an extended class $\mathcal{K}_{\infty}$ function $\alpha:\mathbb{R}\rightarrow\mathbb{R}$ such that
    \begin{align}\label{eqn:cbf}
        \max\limits_{\bm{u}\in\mathcal{U}} c_0(\bm{x})+\bm{c}^T(\bm{x})\bm{u} \, \revision{>} \, 0,
    \end{align}
    with $c_0(\bm{x})=\nabla_{\bm{x}}^T \psi (\bm{x})\bm{f}(\bm{x})+\alpha(\psi(\bm{x}))$ and $\bm{c}^T(\bm{x})=[c_1(\bm{x}), \ldots, c_m(\bm{x})] = \nabla_{\bm{x}}^T \psi(\bm{x}) \bm{G}(\bm{x})$
    holds for all $\bm{x}\in \mathcal{C}$, $\psi(\cdot)$ is called \emph{control barrier function} (CBF) and 
    every Lipschitz control law $\bm{\pi}(\cdot)$ such that $\bm{\pi}(\bm{x})\in\{\bm{u}\in\mathcal{U}: c_0(\bm{x})+\bm{c}^T(\bm{x})\bm{u}> 0\}$ renders the system \eqref{eq:si_sys} safe with respect to $\mathcal{C}$.\looseness=-1
\end{lemma}
Since the crucial component necessary for the application of this result is access to a valid CBF, we require the following.\looseness=-1
\begin{assumption}\label{as:CBF}
    A twice differentiable CBF $\psi(\cdot)$ and a Lipschitz continuous extended class $\mathcal{K}_{\infty}$ function $\alpha(\cdot)$ \revision{satisfying} %
    \begin{align*}
        \revision{\forall \bm{x}\in\mathcal{C}: \exists \bm{u}\in\mathcal{U}:\quad c_0(\bm{x})+\bm{c}^T(\bm{x})\bm{u}> 0.}
    \end{align*}%
    are known for the system dynamics \eqref{eq:si_sys} \revision{on the compact set $\mathcal{C}$.}
\end{assumption}
\revision{For arbitrary functions $\bm{f}(\cdot)$, $\bm{G}(\cdot)$, finding a control barrier function satisfying these requirements is generally a challenging problem.} 
However, given some knowledge about the structure of the dynamical system, e.g., internal integrator chains \revision{as commonly found in Euler-Lagrange systems,}
suitable CBFs can be iteratively constructed without knowledge of the \revision{particular functions $\bm{f}(\cdot)$ and $\bm{G}(\cdot)$} using the approach proposed in \cite{Nguyen2016ExponentialCB}. 
\revision{Note that the focus of our work is not the construction of CBFs} but rather their exploitation to ensure safety with event-triggered learning. \revision{Hence, we leave the extension to scenarios with unknown CBFs for future work similarly as related approaches in literature \cite{wu_safe_2022, zhang_gaussian_2025, zhang_real-time_2024, castaneda2023}.}

If we knew $\bm{f}(\cdot)$ and $\bm{G}(\cdot)$, these assumptions would allow us to immediately define a safety filter \cite{Ames2017}
\begin{subequations}\label{eq:qp}
\begin{align}
\bm{u}^*(x)=&\arg \min _{\bm{u} \in \mathbb{R}^m}\left\|\bm{\pi}_{\text {nom }}(\bm{x})-\bm{u}\right\|^2 \\
&\text{s.t. }c_0(\bm{x})+\bm{c}^T(\bm{x})\bm{u}\geq 0\label{eq:cbf_const_base}
\end{align}
\end{subequations}
which minimally modifies a given nominal control law $\bm{\pi}_{\mathrm{nom}}:\mathcal{X}\rightarrow\mathcal{U}$, such that \eqref{eqn:cbf} can be straightforwardly ensured. Under sufficient smoothness of $\alpha(\cdot)$, $\psi(\cdot)$ and $\bm{\pi}_{\mathrm{nom}}(\cdot)$, the resulting control law can even be shown to be locally Lipschitz continuous \cite{Ames2017}, such that \cref{lem:cbf} ensures the safety of the closed-loop system. 
However, it is not directly applicable in a setting where $\bm{f}(\cdot)$ and $\bm{G}(\cdot)$ are unknown. \revision{For example, if we do not know $\bm{G}(\cdot)$, we cannot even determine the signs of $c_i(\bm{x})$, $i=1,\ldots,m$, and do not know which directions are safe for $\bm{u}$.}\looseness=-1

To overcome this limitation, we assume the ability to collect data online such that we can learn approximations of the safety conditions \eqref{eq:cbf_const_base}.
For this data, we require the following.
\begin{assumption}\label{as:data}
    At arbitrary sampling times $t\in\real_{\ge 0}$, training data $([\bm{x}^T(t),\bm{u}(t)]^T,y(t))$ with outputs 
    \revision{$y(t)=\frac{\mathrm{d}}{\mathrm{d}t}\psi(\bm{x}(t))+\omega_t$}
    perturbed by i.i.d. Gaussian noise $\omega_t\sim\mathcal{N}(0,\sigma_{\mathrm{on}}^2)$, $\sigma_{\mathrm{on}}^2\in\mathbb\real_{>0}$, can be collected.
\end{assumption}

To enable event-triggered learning strategies, it is common to require sampling at arbitrary times \cite{Umlauft2019FeedbackLB}. \revision{In practice, this requirement is reasonably satisfied by regular sampling with a sufficiently high frequency.}
\cref{as:data} requires noise-free state measurements, while training targets $y$ can be perturbed by Gaussian noise. This is a frequently used assumption in the GP-based control literature~\cite{Jagtap2020ControlBF, zhang_event-triggered_2024} \revision{as it allows to model inaccuracies in data,  e.g., from numerical differentiation.}
\revision{The restriction to a noise scenario is necessary to quantify learning errors, but extensions to other noise distributions \cite{reed_error_2025}, non-i.i.d. data \cite{Chowdhury2017a}, and noisy state measurements \cite{dissertation_armin} are possible.}

Based on this system description and Assumptions \ref{as:CBF} and \ref{as:data}, we consider the problem of designing a sampling strategy to determine when \revision{and which} measurements need to be taken, such that safety as defined in \cref{def:safety} is guaranteed.

\section{Safe Control through Event-Triggered Learning}\label{sec:con-time}

We introduce the fundamentals of Gaussian process regression in Section \ref{subsec:GPR}. In \cref{subsec:GP_affine}, we demonstrate how control-affine system models can be inferred from data along with prediction error bounds. Based on these GP models, we present our approach for synthesizing safe control laws in \cref{subsec:synthesis}. We use this control law as a basis for \revision{the design of a safe excitation signal in \cref{subsec:trlearning}. In \cref{eq:event_trig_learn}, we prove the probabilistic safety of our event-triggered learning algorithm that switches to the excitation signal for data generation.} 
Guarantees for the exclusion of Zeno behavior with the proposed event trigger are given in Section~\ref{sec:zeno}.

\subsection{Gaussian Process Regression}
\label{subsec:GPR}
Gaussian process regression (GPR) is a statistical method based on the concept that any finite number of measurements $\{h(\bm{q}^{(1)}),\ldots,h(\bm{q}^{(N)})\}$, $N\in\mathbb{N}$, of an unknown function $h:\mathcal{Q}\rightarrow\mathbb{R}$ evaluated at 
$\bm{q}\in\mathcal{Q}$ from some index set $\mathcal{Q}$, e.g., $\mathcal{Q}=\mathbb{R}^n$, follows a joint Gaussian distribution. A GP, denoted $h(\cdot)\sim \mathcal{GP}(\hat{h}(\cdot),k_h(\cdot,\cdot)) $, is fully specified using a prior mean $\hat{h}:\mathcal{Q}\rightarrow\mathbb{R}$ and a positive definite kernel function $k_h:\mathcal{Q}\times\mathcal{Q}\rightarrow\real_{>0}$  \cite{Rasmussen2006}. The mean function incorporates prior model knowledge, 
which we set to $\hat{h}(\cdot) = 0$ if no prior knowledge about the function is available. This is also assumed in the following without loss of generality. 
The kernel function $k_h(\cdot,\cdot)$ encodes abstract information about the structure of $h(\cdot)$, such as smoothness or periodicity.

Given training data $\mathbb{D} = \{\bm{q}^{(i)},y^{(i)}\}_{i=1}^N$ consisting of $N$ training inputs $\boldsymbol{q}^{(i)} \in \mathcal{Q}$ and noisy measurements $y^{(i)}=h(\bm{q}^{(i)})+\omega^{(i)}$, $\omega^{(i)}\sim\mathcal{N}(0,\sigma_{\mathrm{on}}^2)$, $\sigma_{\mathrm{on}}^2\in \real_{>0}$, we can compute the posterior GP by conditioning the prior on $\mathbb{D}$. The resulting posterior is Gaussian with mean and variance defined by
\begin{subequations}
    \begin{align*}
\mu(\bm{q}) &=\bm{k}_h^{T}(\bm{q})\left(\bm{K}_h+\sigma_{\mathrm{on }}^{2} \bm{I}_{N}\right)^{-1} \bm{y}, \\
\sigma^{2}(\bm{q}) &=
k_h(\bm{q}, \bm{q})-\bm{k}_h^{T}(\bm{q})\left(\bm{K}_h+\sigma_{\mathrm{on}}^{2} \bm{I}_{N}\right)^{-1} \bm{k}_h(\bm{q}),
\end{align*}
\end{subequations}
where $\bm{k}_h(\bm{q})$ and $\bm{K}_h$ are defined element-wise via $k_{h,i}(\bm{q})=k_h(\bm{q},\bm{q}^{(i)})$ and $K_{h,ij}=k(\bm{q}^{(i)},\bm{q}^{(j)})$, respectively, and $\bm{y}=[y^{(1)}\ \cdots\ y^{(N)}]^T$. 

\subsection{Learning Models of Control-Affine Systems}
\label{subsec:GP_affine}
\revision{Since the left-hand side in \eqref{eq:cbf_const_base} is a control-affine function,} our learned system model should exploit this knowledge and provide a model of the same structure. Including this \revision{information} can be straightforwardly achieved with GPs by employing composite kernels for regression \cite{Lederer2021HowTD}. For this purpose, we define a GP prior for \revision{$c_i(\cdot)$, $i=0,\ldots,m$ in \eqref{eq:cbf_const_base}} such that\looseness=-1
    \begin{align}
    c_i(\cdot)\sim \mathcal{G P}\left(0, k_{i}\left(\cdot, \cdot\right)\right).
    \label{eq:gprior}
\end{align}
This implies that the composite prior $h(\cdot)\sim\mathcal{GP}(0, k_h(\cdot,\cdot))$ for $h(\bm{q})=c_0(\bm{x})+\bm{c}^T(\bm{x})\bm{u}$ with $\bm{q}=[\bm{x}^T,\bm{u}^T]^T$ is defined via the composite kernel \cite{Lederer2021HowTD},\looseness=-1
\begin{align*}
    k_h(\bm{q},\bm{q})&=k_0(\bm{x},\bm{x}')+ \sum\limits_{i=1}^{m} u_i k_{i}(\bm{x},\bm{x}')u_i'.
\end{align*}
Using these priors, it is straightforward to derive the posterior distributions of the functions $c_i(\cdot)$, $i=0,\ldots,m$ analogously to standard GPR by conditioning the joint prior of the individual functions on the training data. 
The resulting posterior distributions are again Gaussian with means and variances \cite{Lederer2021HowTD}
\begin{subequations}
    \begin{align}
\!\mu_{0}(\bm{x})&=\bm{k}_{0}^{T}(\bm{x})\left(\bm{K}_h+\sigma_{\mathrm{on}}^{2} \bm{I}_{N}\right)^{-1} \bm{y},
\label{eq:muf}
\\
\!\mu_{i}(\bm{x})&=\bm{k}_{i}^{T}(\bm{x}) \bm{U}_i\left(\bm{K}_h+\sigma_{\mathrm{on}}^{2} \bm{I}_{N}\right)^{-1} \bm{y}
\label{eq:mug}
\\
\!\sigma_{0}^{2}(\bm{x})&=k_{0}(\bm{x}\!, \bm{x})\!-\!\bm{k}_{0}^{T}\!(\bm{x}) \left(\bm{K}_h\!+\!\sigma_{\mathrm{on}}^{2} \bm{I}_{N}\right)^{-1}\! \bm{k}_{0}(\bm{x}),\label{eq:sigmaf}
\\
\!\sigma_{i}^{2}(\bm{x})&=k_{i}(\bm{x},\! \bm{x})\!-\!\bm{k}_{i}^{T}\!(\bm{x}) \bm{U}_i\!\left(\bm{K}_h\!+\!\sigma_{\mathrm{on}}^{2} \bm{I}_{N}\right)^{-1} \bm{U}_i \bm{k}_{i}(\bm{x})\!\label{eq:sigmag}
\end{align}
\end{subequations}
for $i=1,\ldots,m$, where $\bm{U}_i=\operatorname{diag}([u_i^{(1)} \ldots u_i^{(N)}])$ and $\bm{k}_i(\bm{x})$ are defined analogously to $\bm{k}_h(\bm{x})$. In order to exploit the Bayesian foundations of GP regression for the derivation of prediction error bounds, we make the following assumption.
\begin{assumption}\label{as:prior}
    The unknown functions $c_i(\cdot)$, $i=1,\ldots,m$, are samples from prior GPs \eqref{eq:gprior} defined using stationary kernels $k_i(\cdot,\cdot)$, i.e., they are functions of the difference of their arguments satisfying $k_i(\bm{x},\bm{x})=s_i^2$ for all $\bm{x}\in\mathcal{X}$ with signal variances $s_i^2\in\mathbb{R}_{\geq0}$. Moreover, the kernels have continuous partial derivatives up to the fourth order. 
\end{assumption}
This assumption of suitable prior distributions is commonly employed when working with Bayesian models, see e.g., \cite{Dhiman2021CB, wang_safe_2018}. While it effectively limits the admissible class of unknown functions to the sample space of the GP prior, this restriction is often not severe, in particular when working with universal kernels \revision{capable of approximating continuous functions arbitrarily well \cite{Rasmussen2006, Micchelli2006UniversalK}. The additional restrictions on the kernels, i.e., stationarity and a sufficient smoothness are generally not restrictive, and they are satisfied, e.g., by the frequently used squared exponential kernel, which are also universal.}
Therefore, \cref{as:prior} does not pose a significant limitation. 
Based on this assumption, we introduce next prediction error bounds for $\mu_0(\cdot)$ and $\mu_{i}(\cdot)$.\revision{\footnote{\revision{Proofs of all theoretical results in this section appear in the Appendix.}}}
\begin{lemma}\label{lem:GPbound}
    Consider unknown functions $\{c_i(\cdot)\}_{i=0}^m$ \revision{and training data satisfying Assumptions \ref{as:data} and \ref{as:prior}. For $i=0,\ldots,m$, let $\underline{\sigma}_i\in\mathbb{R}_{\geq0}$, $\delta\in(0,1)$ be parameters and define
    \begin{align}\label{eq:beta_i}
        \beta_i(\delta, \underline{\sigma}_i)&= 8\log\bigg(\frac{2(m+1)}{\delta}\prod\limits_{j=1}^n\Big(1+ \max\limits_{\bm{x}\in\mathcal{X}}x_j\!-\!\min\limits_{\bm{x}\in\mathcal{X}}x_j\Big) \bigg) 
        \nonumber \\
        & \quad + 8n \log\bigg(1+\frac{\sqrt{n}(L_i+L_{\mu_i}+L_{\sigma_i})}{2\underline{\sigma}_i}\bigg),
    \end{align}
    where $L_{\mu_i}$, $L_{\sigma_i}$ are Lipschitz constants of $\mu_i(\cdot)$, $\sigma_i(\cdot)$, and $L_i$ are as specified in \cite[Theorem 3.2]{Lederer2019UniformEB}.
    }
    Then, the prediction error of GP regression is jointly bounded by 
            \begin{align}
        |\mu_{i}(\bm{x})-c_i(\bm{x})|&\leq \revision{\sqrt{\beta_{i}(\delta, \underline{\sigma}_i)} \tilde{\sigma}_i(\bm{x})},
    \label{eq:GPboundg}
    \end{align}
    for all $\bm{x}\in\mathcal{X}$, $i=0,\ldots,m$ with probability of at least $1-\delta$ \revision{with $\tilde{\sigma}_i(\bm{x})=\max\{\sigma_i(\bm{x}),\underline{\sigma}_i \},$ $i=0,\ldots,m$.}
\end{lemma}

By exploiting the Bayesian foundation of GPs, this result provides us with probabilistic uniform prediction error bounds for $c_i(\cdot)$, $i=0,\ldots,m$, individually. The error bounds in the right-hand side of~\eqref{eq:GPboundg} have only $\underline{\sigma}_i$ as design parameters, which have an intuitive interpretation: they \revision{allow us to specify} a lower bound on the certifiable prediction error. 
\subsection{Safe Learning-Based Control Synthesis}
\label{subsec:synthesis}
While we assume to not have access to the functions $\{c_i(\cdot)\}_{i=0}^m$, GP regression allows us to infer models in the form of $\mu_i(\cdot)$ from training data. Since these models come along with error bounds, cf. \cref{lem:GPbound}, we robustify the CBF condition \eqref{eq:cbf_const_base} to account for model uncertainty
\revision{and formulate our control law as follows
\begin{subequations}\label{eq:opt_QP}
\begin{align}
\bm{\pi}(\bm{x})=&\min_{\bm{u}} ~  \left\|\bm{u}-\bm{\pi}_{\mathrm{nom}}(\bm{x})\right\| \\
  &\text{s.t.}~\hat{c}_0(\bm{x})+\hat{\bm{c}}^T(\bm{x})\bm{u}-\|\bm{Q}(\bm{x})\bm{u}\|\geq 0,  \label{eq:cs_cons}
\end{align}
\end{subequations}
where 
\begin{align*}
    \hat{c}_0(\bm{x})&=\mu_0(\bm{x})-\sqrt{\beta_0(\delta, \underline{\sigma}_0)}\tilde{\sigma}_0(\bm{x}),\\
    \hat{\bm{c}}(\bm{x})&=\begin{bmatrix}
        \mu_1(\bm{x})&\cdots&\mu_m(\bm{x})
    \end{bmatrix}^T,\\
    \bm{Q}(\bm{x})&=\sqrt{m}\mathrm{diag}(\sqrt{\beta_1(\delta, \underline{\sigma}_1)}\tilde{\sigma}_1(\bm{x}),\ldots,\sqrt{\beta_m(\delta, \underline{\sigma}_m)}\tilde{\sigma}_m(\bm{x})).
\end{align*}
Note that our formulation resembles the structure found in related work \cite{castaneda2023, zhang_real-time_2024, Long2022SafeCS}. In fact, the optimization problem \eqref{eq:opt_QP} can be reformulated as a SOCP \cite{castaneda2023}. However, it differs in the expression for $\bm{Q}(\bm{x})$, for which we exploit the individual error bound in \cref{lem:GPbound}. This crucial difference allows us a simplification of feasibility conditions as shown in the following result.
}\looseness=-1
\begin{proposition}\label{prop:safety}
Consider a system \eqref{eq:si_sys}, GP priors \eqref{eq:gprior} and a fixed data set $\mathbb{D}$ such that Assumptions  
\ref{as:CBF} - \ref{as:prior} are satisfied. \revision{Assume that for all $\bm{x}\in\mathcal{C}$, there exists $i=1,\ldots,m$ such that 
\begin{align}\label{eq:GP_cond}
    |\mu_i(\bm{x})|> \sqrt{m\beta_i(\delta, \underline{\sigma}_i)}\tilde{\sigma}_i(\bm{x}).
\end{align}}
Then, \eqref{eq:opt_QP} is feasible and ensures safety for all $t\in\mathbb{R}_{\geq0}$ with probability $1-\delta$.
\end{proposition}

Proposition~\ref{prop:safety} allows us the straightforward design of a safe control law, even though we only have access to the learned models $\mu_i(\cdot)$. For achieving this, 
we merely need to check condition
\eqref{eq:GP_cond}, which essentially requires a sufficiently small posterior standard deviation $\sigma_i(\bm{x})$ for all $\bm{x}\in\mathcal{C}$ together with a sufficiently small value $\underline{\sigma}_i$. \revision{Note that the safety filter \eqref{eq:opt_QP} generally does not preserve stability guarantees of the nominal control law $\bm{\pi}_{\mathrm{nom}}(\cdot)$ similar to related work \cite{Jagtap2020ControlBF, zhang_gaussian_2025, zhang_real-time_2024, castaneda2023}, but the additional consideration of a suitable stability constraint can resolve this weakness \cite{Mestres2023FeasibilityAR}.}

\subsection{\revision{Excitation Requirement for Guaranteed Accuracy}}
\label{subsec:trlearning}
While \cref{prop:safety} provides a way 
to ensure safety using learned GP models, it requires that \eqref{eq:GP_cond} is satisfied for all $\bm{x}\in\mathcal{X}$. This assumption can be challenging to ensure and formally showing that it holds requires determining the global minimum of the left-hand side of \eqref{eq:GP_cond}, which is a non-convex function in general.
Thus, an intuitive approach is to update the GP model online to improve its accuracy.
\revision{By choosing a suitable control input for data collection, this approach ensures} that \eqref{eq:GP_cond} continues to be \revision{probabilistically} satisfied \revision{for at least one $i=1,\ldots,m$} after an update \revision{without compromising safety. This is formalized in the following result.}\looseness=-1

\begin{proposition}\label{prop:learn_policy}
Consider a system \eqref{eq:si_sys}, GP priors \eqref{eq:gprior} and a fixed data set $\mathbb{D}$ with $N\in\mathbb{N}$ data points such that Assumptions  
\ref{as:CBF} - \ref{as:prior} are satisfied.
\revision{Given $i \in \{1,\ldots,m\}$, define the control law 
$\bar{\bm{\pi}}_i:\mathbb{R}^n\rightarrow\mathbb{R}^m$ via $\bar{\pi}_{i,j}(\bm{x})=0$ for $j\neq i$ and
\begin{align}\label{eq:opt_trig}
    \bar{\pi}_{i,i}(\bm{x})&=\begin{cases}
        \mathrm{sgn}(\xi_i(\bm{x}))\bar{u}_{\mathrm{GP}}(\bm{x}) &\text{if } \frac{-\hat{c}_0(\bm{x})}{|\xi_i(\bm{x})|}\leq \bar{u}_{\mathrm{GP}}(\bm{x})\\
        \frac{|\hat{c}_0(\bm{x})|}{\xi_i(\bm{x})} &\text{else}
    \end{cases}
\end{align}
where}
\begin{align}
    \revision{\xi_i(\bm{x})}&\revision{=\mu_i(\bm{x})-\mathrm{sgn}(\mu_i(\bm{x}))\sqrt{m\beta_i(\delta, \underline{\sigma}_i)}\tilde{\sigma}_i(\bm{x})}\\
    \bar{u}_{\mathrm{GP}}(\bm{x})&=\sqrt{ \frac{(1\!+\!\epsilon\!+\!\gamma\!+\!\revision{\frac{1}{\sqrt{m}}})^2\revision{m}\beta_i(\delta,\underline{\sigma}_i^{+})s_i^2(s_0^2\!+\!s_\mathrm{on}^2)}{\revision{\epsilon^2 \beta_i(\delta,\underline{\sigma}_i)\tilde{\sigma}_i^2(\bm{x})}}}\label{eq:underline_u_GP}
\end{align}
\revision{and $\underline{\sigma}_i^{+}\in\mathbb{R}_{>0}$ is defined such that 
\begin{align}\label{eq:sigma_under_cond}
    \sqrt{\beta_i(\delta, \underline{\sigma}_i^{+})} \underline{\sigma}_i^{+}\leq \frac{\epsilon\sqrt{\beta_i(\delta, \underline{\sigma}_i)} \tilde{\sigma}_i(\bm{x})}{\sqrt{m}(1+\epsilon+\gamma)+1}
\end{align}
holds for some $\epsilon,\gamma\in\mathbb{R}_{>0}$ and $0<\underline{\sigma}_i\leq s_i^2=k_i(\bm{x},\bm{x})$.
If
\begin{align}\label{eq:prior_safety_margin}
    \frac{|\mu_i(\bm{x})|}{\sqrt{m\beta_i(\delta, \underline{\sigma}_i)}\tilde{\sigma}_i(\bm{x})}\geq 1+\epsilon 
\end{align}
holds at $\bm{x}\in\mathcal{C}$, then, $\bar{\bm{\pi}}_i(\cdot)$
satisfies \eqref{eq:cbf_const_base} at $\bm{x}$ with probability $1-\delta$.} Moreover, the GP models with mean $\mu_i^+(\cdot)$ and standard deviation $\sigma_i^+(\cdot)$ trained using an updated data set $\mathbb{D}\cup \{[\bm{x},\bar{\bm{\pi}}_i(\bm{x})]^T,y\}$ guarantees with probability $1-\delta$ that
\begin{align}\label{eq:cond_event}
    \frac{|\mu_i^+(\bm{x})|}{\sqrt{\beta_i(\delta,\underline{\sigma}_i^{+})}\tilde{\sigma}_i^+(\bm{x})} \geq 1 + \epsilon + \gamma.
\end{align}
\end{proposition}
\revision{\cref{prop:learn_policy}} requires condition \eqref{eq:sigma_under_cond}
for $\underline{\sigma}_i^{\revision{+}}$ to probabilistically enable a sufficient reduction of $\tilde{\sigma}_i^+(\cdot)$ through a training point. 
\revision{Note that a positive value $\underline{\sigma}_i^{+}$ satisfying \eqref{eq:sigma_under_cond} is guaranteed to exist due to the logarithmic dependency of $\beta(\cdot,\cdot)$ on this parameter, such that $\beta(\delta,\underline{\sigma}_i^{+})$ is finite.}
\revision{Condition \eqref{eq:prior_safety_margin} probabilistically guarantees that a safe control direction exists taking the role of \eqref{eq:GP_cond} in \cref{prop:safety}. These properties are exploited in the excitation filter \eqref{eq:opt_trig}} 
by restricting the control input to amplitudes \revision{$|\bar{\pi}_{i,i}(\bm{x})|\geq \bar{u}_{\mathrm{GP}}(\bm{x})$, such that enough information about $c_i(\cdot)$ can be extracted to safely ensure \eqref{eq:cond_event} with probability $1-\delta$.} 
\revision{Since \eqref{eq:GP_cond} only needs to be satisfied for at least one $i=1,\ldots,m$, this approach increases the overall safety margin $\max_{i=1,\ldots,m}\nicefrac{|\mu_i(\bm{x})|}{\sqrt{\beta_i(\delta,\underline{\sigma}_i}\bar{\sigma}_i(\bm{x})}-1$.}
Importantly, the \revision{required control amplitude} $\bar{u}_{\mathrm{GP}}(\bm{x})$ can easily be computed and depends on \revision{positive, finite values.} 
\revision{Hence, \cref{prop:learn_policy} implies the sufficiency of finite control inputs.}\looseness=-1

\subsection{Safe Event-Tiggered Gaussian Process Learning}\label{eq:event_trig_learn}
Since the excitation filter generally results in a worse closed-loop performance, it is desirable to use $\bar{\bm{\pi}}_i(\cdot)$ only when we collect training samples. Therefore, we propose to distinguish two phases with different goals in our overall control approach:
\begin{itemize}
    \item When we are at risk of violating \eqref{eq:GP_cond}, we focus on improving the model accuracy and employ the excitation filter \eqref{eq:opt_trig} to generate informative training samples. We determine the necessity of GP model updates through an event trigger. Given $N\in\mathbb{N}$, 
    \begin{subequations}\label{eq:event_trigger}
    \begin{align}
        t_{N+1}=&\inf\limits_{t>t_N} t\\
        &\text{s. t. } \revision{\max_{j=1,\ldots,m}}\frac{\mu_j(\bm{x}(t))}{\sqrt{m\beta_j(\delta,\underline{\sigma}_j)}\tilde{\sigma}_j(\bm{x}(t))} \leq 1 \!+\! \epsilon,\label{eq:trigger}
    \end{align}
    \end{subequations}
    which is used to activate the excitation filter \eqref{eq:opt_trig}, i.e., $\bm{u}(t) =\bar{\bm{\pi}}_i(\bm{x}(t))$ at $t=t_{N}$, 
    \revision{with $i$  the maximizer in \eqref{eq:trigger}.}
    Therefore, new training samples have the form $([\bm{x}^T(t_{N+1}),\bar{\bm{\pi}}_i(\bm{x}^T(t_{N+1}))]^T,y)$. 
    \item When the model accuracy is sufficient, i.e., \eqref{eq:trigger} is not satisfied, we focus on control performance and directly apply the safe policy \eqref{eq:opt_QP} 
    which minimizes the deviation from the nominal control law $\bm{\pi}_{\mathrm{nom}}(\cdot)$. Thus, we choose the control signal in this phase as $\bm{u}(t)=\bm{\pi}(\bm{x}(t))$ for all $t\in(t_N,t_{N+1}), N\in\mathbb{N}$.
\end{itemize}
This event-triggered learning approach, which is summarized in \cref{alg:control_algorithm}, \revision{probabilistically} guarantees safety.

\begin{algorithm}[t]
    \begin{algorithmic}[1]
        \State $N\gets 1$
        \State \revision{set confidence level $\tilde{\delta}\in(0,1)$}
        \State \revision{pick $\underline{\sigma}_j\leq s_j$, for all $j=0,\ldots,m$}
        \While{\textbf{true}}
            \If{$\revision{\max_{j=1,\ldots,m}}\frac{\mu_{\revision{j}}(\bm{x}(t))}{\sqrt{\revision{m\beta_j(\nicefrac{6\tilde{\delta}}{\pi^2N^2},\underline{\sigma}_j)}}\tilde{\sigma}_{\revision{j}}(\bm{x}(t))} = 1 + \epsilon$}
                \State \revision{set $i$ as maximizer of \eqref{eq:trigger}}
                \State \revision{determine $\underline{\sigma}_i^{+}$ according to \eqref{eq:sigma_under_cond}}
                \State apply $\bm{u}=\bar{\bm{\pi}}_{\revision{i}}(\bm{x})$ with $\bar{\bm{\pi}}_{\revision{i}}(\cdot)$ defined in \eqref{eq:opt_trig}
                \State measure $y=\revision{c_0(\bm{x}(t_N))+\bm{c}^T(\bm{x}(t_N))\bm{u}}+\omega$
                \State $\mathbb{D}\gets \mathbb{D}\cup \{ ([\bm{x}^T(t_N), \bm{u}]^T,y \}$
                \State update $\mu_{\revision{j}}(\cdot)$, $\sigma_{\revision{j}}(\cdot)$ \revision{for all $j=0,\ldots, m$  }
                \State $t_N\gets t$, $N\gets N+1$, $\underline{\sigma}_i\gets\underline{\sigma}_i^{+}$
            \Else
                \State \revision{compute safety-filtered control $\bm{\pi}(\bm{x})$ using \eqref{eq:opt_QP}}
                \State apply $\bm{u}=\bm{\pi}(\bm{x})$
            \EndIf   
        \EndWhile
    \end{algorithmic}
    \caption{\revision{Safe Event-Triggered Learning for Control}}\label{alg:control_algorithm}
\end{algorithm}
\begin{theorem}\label{th:safe_learn}
    Consider a system \eqref{eq:si_sys} and GP priors \eqref{eq:gprior} such that Assumptions  \ref{as:CBF} - \ref{as:prior} are satisfied. 
    If \eqref{eq:prior_safety_margin} holds at $t=t_0$ and $\bm{x}(t_0)\in\mathcal{C}$, \revision{executing} \cref{alg:control_algorithm} \revision{with confidence $\tilde{\delta}\in(0,1)$} yields a well-defined control law and guarantees safety during the time interval \revision{$[t_0,t_{\bar{N}})$} with probability $1-\tilde{\delta}$ \revision{for all $\bar{N}\in\mathbb{N}$.}
\end{theorem}

This result guarantees that the system stays in the safe set \revision{for an arbitrarily high number of GP updates $N\in\mathbb{N}$ with probability $1-\tilde{\delta}$} if it starts inside it at a point where \eqref{eq:prior_safety_margin} is satisfied. \revision{Providing a probabilistic guarantee over multiple GP updates is achieved by decreasing the admissible violation probability $\delta$ at a rate of $\nicefrac{6\tilde{\delta}}{\pi^2N^2}$, such the union bound over all $N\in\mathbb{N}$ bounds the joint violation probability as the sum of individual ones, which yields $\tilde{\delta}$.}
The required condition \revision{\eqref{eq:prior_safety_margin} for the initial state} is necessary to prevent any safety violations from the control before any safe learning is possible. This is easily satisfied at the initial condition $\bm{x}(0)$ by initializing the GP model with a suitable prior or by providing training data obtained a priori, e.g., by running a locally safe controller.

\subsection{Ruling out Zeno Behavior}\label{sec:zeno}
Avoiding Zeno behavior is crucial in event-triggered learning with GPs to prevent updates from accumulating in a finite amount of time, which would make a practical implementation impossible. In order to exclude Zeno behavior in event-triggered learning, it is sufficient to uniformly lower bound the time between any two consecutive events by a positive constant. This approach is employed in the following result.
\begin{proposition}\label{prop:zeno}
    Consider a system \eqref{eq:si_sys} and GP priors \eqref{eq:gprior} such that Assumptions  \ref{as:CBF} - \ref{as:prior} are satisfied. 
    If \eqref{eq:prior_safety_margin} holds at $t=t_0$ and $\bm{x}(t_0)\in\mathcal{C}$, \cref{alg:control_algorithm} triggers model updates with an inter-event time $\Delta_N=t_{N+1}-t_N$ satisfying $\Delta_N\geq \nicefrac{\gamma}{L_{\Gamma_N}}$ for all $N\in\mathbb{N}$ with probability $1-\tilde{\delta}$, where $L_{\Gamma_N}$ denotes the Lipschitz constant of the triggering function 
    \begin{equation*}
        \Gamma_N (t) = \revision{\max_{j=1,\ldots,m}\frac{\mu_j(\bm{x}(t))}{\sqrt{m\beta_j(\nicefrac{6\tilde{\delta}}{\pi^2N},\underline{\sigma}_j)}\tilde{\sigma}_j(\bm{x}(t))}.}
    \end{equation*}
\end{proposition}

This result exploits the Lipschitzness of the state trajectory together with the GP means $\mu_i(\cdot)$ and the standard deviations $\sigma_i(\cdot)$ to obtain a Lipschitz constant $L_{\Gamma}$. Since update events are triggered with a higher threshold than \revision{probabilistically} ensured after the update, the Lipschitz continuity of $\Gamma_N(\cdot)$ directly implies a lower bound on the inter-event times. The constant $L_{\Gamma}$ captures the change rate of the system and the dependency on the GP prior. Intuitively, when the system or the GP model exhibit a high variability, $L_{\Gamma}$ is large, so that a high triggering frequency can occur. This can be compensated by increasing the gap between the trigger condition and the update objective, which is given by the constant $\gamma$. \revision{This parameter directly influences the magnitude of control inputs of the excitation filter due to \eqref{eq:underline_u_GP}.} Hence, it allows us to effectively tune the inter-event times of \cref{alg:control_algorithm} \revision{via the excitation control signal.}\looseness=-1

\section{Numerical Evaluation in Adaptive Cruise Control} \label{sec:numerical_eval}
\definecolor{azrablue}{RGB}{0, 103, 180}
\definecolor{ayellow}{RGB}{255, 195, 0}
\definecolor{orang}{RGB}{160, 50, 50}
\definecolor{gr}{RGB}{85, 139, 47}

We demonstrate the effectiveness of our framework by considering the example of adaptive cruise control, \revision{which we model using $\bm{f}(\bm{x})=\begin{bmatrix}
    -\nicefrac{F_r(v)}{m}& v_0-v
\end{bmatrix}^T$, and $\bm{G}(\bm{x})=\begin{bmatrix}
    0 & 1/m
\end{bmatrix}^T$,}
where the state $\bm{x} \!= \!\begin{bmatrix}
 v \! & z  
\end{bmatrix}^T \! \in \mathbb{R}^2$ is composed of the distance to the front vehicle $z$ and the ego vehicle velocity $v$. The parameter $m\!=\!1650$ corresponds to the ego vehicle's mass, $v_0$ is the front vehicle's velocity and $F_r(\tilde{v})=f_0+f_1 (\tilde{v}+v_0)+f_2 (\tilde{v}+v_0)^2$ is the rolling resistance force on the ego vehicle with parameters $f_0=0.2$, $f_1=10$ and $f_2=0.5$. 
The objective is to reach the desired velocity $v_d$, for which we design a nominal velocity controller $\pi_{\mathrm{nom}}(\bm{x})=-20(v-v_d)$. Since it is crucial that no collision with the front vehicle occurs, we define a CBF $\psi(\bm{x}) = z - T_h v,$
where $T_h=1.8$ is the lookahead time and use $\alpha(\psi)= 65 \psi$, \revision{such that \cref{as:CBF} holds}. As we assume that the functions $c_0(\cdot)$ and $c_1(\cdot)$ are unknown, we model their behavior by putting a prior GP distribution $\mathcal{GP}\left(0, k_{\revision{0}}(\cdot, \cdot)\right)$ and $\mathcal{GP}\left(0, k_{\revision{1}}(\cdot, \cdot)\right)$ on them. For $k_{\revision{0}}(\cdot, \cdot)$ and $k_{\revision{1}}(\cdot, \cdot)$, we employ the squared exponential kernel, whose hyperparameters are set to $l_{\revision{0}} = 1$, $l_{\revision{1}} = 2$, $s_{\revision{0}} =1$, $s_{\revision{1}} =0.5$. \revision{This choice aligns with \cref{as:prior} as these kernels are universal \cite{Micchelli2006UniversalK} and $c_i(\cdot)$ are analytic functions.} In order to comply with the required feasibility of the CBF condition at $t=0$ in \cref{th:safe_learn}, we initialize the composite GP model \eqref{eq:muf} - \eqref{eq:sigmag} with one training point before starting \cref{alg:control_algorithm}. All training samples that we obtain through our proposed framework are perturbed by Gaussian noise with standard deviation $\sigma_{\mathrm{on}} = 0.01$, \revision{such that \cref{as:data} is satisfied}. 
Furthermore, we choose \revision{$\delta=0.01$ and  $\underline{\sigma}_{0,1}=0.01$ for GP error bounds in \eqref{eq:beta_i}. }
When executing \cref{alg:control_algorithm}, we 
choose $\epsilon = 0.2$ and $\gamma = 0.5$ for the design parameters. To demonstrate the effectiveness of the proposed safe control approach through event-triggered learning, we apply \cref{alg:control_algorithm} in a setting with conflicting goals of the nominal controller $\pi_{\mathrm{nom}}(\cdot)$ and the safety conditions by setting $v_0=14$ and $v_d=24$, i.e., the desired velocity $v_d$ is higher than the front vehicle's velocity $v_0$.\looseness=-1

\begin{figure}
    \centering
    \definecolor{nicegreen}{RGB}{0,200,0} 
    \pgfplotsset{width=10\columnwidth/10, compat=1.13,
        height=43\columnwidth/100, grid=major,
        legend cell align=left, ticklabel style={font=\small\sffamily},
        every axis label/.append style={font=\small\sffamily},
        legend style={font=\small\sffamily}, title style={yshift=-7pt, font=\small\sffamily}}
    \centering
    \begin{tikzpicture}
      \node[name=plotvirtual] at (0,0) {
        \begin{tikzpicture}
          \begin{axis}[
              name=plot1,
              height=3.6cm, width=\columnwidth,
              grid=none,
              xmin=0, xmax=30,
              ymin=12, ymax=26,
              ylabel={$v$},
              legend columns=1,
              legend style={font=\footnotesize\sffamily, legend cell align=left, align=left, draw=white!15!black},
              xticklabels={,,}
            ]
            \addplot[thick,azrablue,forget plot] table[x=t,y=x]{Plots11072025/v_11072025.txt};
            \addplot[thick,ayellow,dashed,forget plot] table[x=t,y=x]{Plots08072025/v_real.txt};
            \addplot[thick,amethyst,dash dot,forget plot] table[x=t,y=x]{Plots08072025/v_time_tr_new_.txt};
            \addplot[red,thick,dashed,forget plot] coordinates{(0,24) (30,24)};
            \addplot[gr,line width=1.2pt,dashed,forget plot] coordinates{(0,14) (30,14)};
            \addlegendimage{thick,red,dashed}
            \addlegendentry{Desired velocity}
            \addlegendimage{gr,line width=1.2pt,dashed}
            \addlegendentry{Front vehicle velocity}
          \end{axis}
          \begin{axis}[
              name=plot2,
              at=(plot1.east), anchor=east, yshift=-2.15cm,
              height=3.6cm, width=\columnwidth,
              grid=none,
              xmin=0, xmax=30,
              ymin=20, ymax=105,
              xlabel={$t$},
              ylabel={$z$},
              xlabel shift={-3pt},
              legend columns=1,
              legend style={font=\footnotesize\sffamily, legend cell align=left, align=left, draw=white!15!black}
            ]
            \addplot[thick,azrablue] table[x=t,y=x]{Plots11072025/z_11072025.txt};
            \addplot[thick,ayellow,dashed] table[x=t,y=x]{Plots08072025/z_real.txt};
            \addplot[thick,amethyst,dash dot] table[x=t,y=x]{Plots08072025/z_time_tr_new.txt};
            \addlegendimage{thick,azrablue}
            \addlegendentry{CBF w. GP model}
            \addlegendimage{thick,ayellow,dashed}
            \addlegendentry{CBF w. exact model}
            \addlegendimage{thick,amethyst,dash dot}
            \addlegendentry{Time-triggered learning}
          \end{axis}
        \end{tikzpicture}
      };
    \end{tikzpicture}
    \caption{
    The continuous-time control approach in \cref{alg:control_algorithm} with event-triggered online learning results in almost identical velocity (top) and distance (bottom) trajectories as a CBF-based controller with exact model knowledge in contrast to an analogous approach with periodically updated GP model, which becomes infeasible after $\approx 6\,$s.
    }
    \label{fig:trajectories_cont}
\end{figure}

\begin{figure}[t!]
    \centering
    \definecolor{nicegreen}{RGB}{0,200,0} 
    \pgfplotsset{width=10\columnwidth /10, compat = 1.13, 
        height = 43\columnwidth /100, grid= major, 
        legend cell align = left, ticklabel style = {font=\small\sffamily},
        every axis label/.append style={font=\small\sffamily},
        legend style = {font=\small\sffamily},title style={yshift=-7pt, font = \small\sffamily} }

    \centering
    \begin{tikzpicture} 
    \node[name=plotvirtual] at (0,0) {\begin{tikzpicture}
    \begin{axis}[
    name=plot1,
    height = 3.6cm, width =\columnwidth,
    grid=none,
    xmin=0, xmax=30,
    ymin=-30000, ymax=30000,
    ylabel={$u$},
    xlabel={$t$},
    xlabel shift={-3pt},
    legend columns=1,
    legend style={font=\footnotesize\sffamily , legend cell align=left, align=left, draw=white!15!black},
    ]

        \addplot[thick, black] table[x=t, y=x] {Plots11072025/us_11072025.txt};
        \addplot[only marks, mark=star, mark size=0.8, mark options={draw=red, fill=red}] table[x=t, y=x] {Plots11072025/us_tr_11072025.txt};
        \addplot[only marks, mark=star, mark size=0.8, mark options={draw=blue, fill=blue}] table[x=t, y=x] {Plots25072025/us_tr_cast_08072025.txt};
        \addlegendentry{Control input}
        \addlegendentry{Update event}
        \addlegendentry{Update event \cite{castaneda2023}}
    \end{axis}
    \end{tikzpicture}};
    \end{tikzpicture}
    \caption{The event-triggered update scheme in \cref{alg:control_algorithm} \revision{probabilistically} ensures the necessary model accuracy for guaranteeing safety by sampling data until the system state has converged. 
    \revision{The excitation filter ensures that data is generated using sufficiently large control amplitudes $\bar{u}_{\mathrm{GP}}(\bm{x})$ with high probability.}
    \revision{When the approach from \cite{castaneda2023} is applied, Zeno behavior occurs and essentially stops the simulation after $\approx 3s$. }
    }
    \label{fig:cbf_u_cont}
\end{figure}

\begin{figure}[t!]
    \centering
    \definecolor{nicegreen}{RGB}{0,200,0} 
    \pgfplotsset{width=10\columnwidth /10, compat = 1.13, 
        height = 43\columnwidth /100, grid= major, 
        legend cell align = left, ticklabel style = {font=\small\sffamily},
        every axis label/.append style={font=\small\sffamily},
        legend style = {font=\small\sffamily},title style={yshift=-7pt, font = \small\sffamily} }
    \centering
    \begin{tikzpicture} 
    \node[name=plotvirtual] at (0,0) {\begin{tikzpicture}
    \begin{axis}[
    name=plot1,
    height = 3.6cm, width =\columnwidth,
    grid=none,
    ymin=16355, ymax=42444,
    xmin=0.027, xmax=0.085,
    ylabel={$\max_t|u(t)|$},
    xlabel={$\min_N t_{N+1}-t_N$},
    legend columns=1,
    xlabel shift={-3pt},
    xtick={0.03,0.04,0.05,0.06,0.07,0.08},
    scaled x ticks=false, 
    xticklabel style={/pgf/number format/fixed},
    legend style={font=\footnotesize\sffamily , legend cell align=left, align=left, draw=white!15!black},
    ]

       \addplot[name path=lower, draw=none]
  table[x=lower_t, y=lower_u]{Plots22072025/stats_merged.txt};
\addplot[name path=upper, draw=none]
  table[x=upper_t, y=upper_u]{Plots22072025/stats_merged.txt};

\addplot[fill=azrablue!30, draw=none]
  fill between[of=lower and upper];

\addplot[thick, black]
  table[x=mean_t, y=mean_u]{Plots22072025/stats_merged.txt};

  \draw[-{Latex}] (axis cs: 0.045, 25000) -- (axis cs: 0.065, 33000);
  \node at (axis cs: 0.054,33000) {$\gamma \uparrow$};

    \end{axis}

    \end{tikzpicture}};
    \end{tikzpicture}
    \caption{\revision{The parameter $\gamma$ allows to trade-off the admissible peak control magnitude $\max_t|u(t)|$ and the minimum inter-event time $\min_N T_{N+1}-t_N$.
    }}
    \label{fig:u_t_max}
\end{figure}

The resulting state trajectories are depicted in \cref{fig:trajectories_cont}. In the top plot, we observe that the velocity of the vehicle steadily approaches $v_d$ until there is a noticeable reduction in distance, as shown in the bottom plot. Due to the CBF, the speed of the vehicle converges to the speed of the ego vehicle $v_0$, while ensuring a safe distance. 
This behavior can be observed independently of the prior availability of exact model knowledge, which is due to our proposed strategy for event-triggered learning strategy. As depicted in 
\cref{fig:cbf_u_cont}, the triggering condition \eqref{eq:event_trigger} generates data at an almost constant rate at the beginning in order to achieve the necessary model accuracy. For ensuring a sufficient information gain with each of these samples, the magnitude of the control input resulting from the CBF-QP \eqref{eq:opt_QP} is adapted in the safe direction. 
Note that \cref{alg:control_algorithm} not only triggers update events when close to the constraint boundary as maintaining the feasibility of the CBF-QP \eqref{eq:opt_QP} is crucial regardless of our distance from the boundary. When the dynamical system has almost reached a stationary point after $t=10$, the triggering stops. Thereby, merely $50$ data points are necessary to \revision{probabilistically} ensure safety using our event-triggered learning approach. This is in strong contrast to a time-triggered online version without excitation filter. It can be clearly seen in
\cref{fig:cbf_u_cont} that periodically updating the GP model cannot ensure feasibility of the CBF-SOCP, which causes a diverging trajectory and constraint violations after $\approx 6s$. \revision{Using the strategy proposed in \cite{castaneda2023} with minimal admissible control amplitudes leads to Zeno behavior after $\approx 3s$, at which point the simulation is essentially brought to a halt and all probabilistic safety guarantees end. This occurrence of Zeno behavior is in stark contrast to our approach, which is designed to probabilistically ensure positive inter-event times for positive values of $\gamma$. 
As illustrated in Figure~\ref{fig:u_t_max}, this parameter establishes a trade-off between the inter-event times and maximal control amplitudes used for excitation, which enables us to tune the sampling behavior of \cref{alg:control_algorithm}.
Thereby, our approach does not only exclude Zeno behavior, but provides some design freedom to account for additional restrictions, e.g., in terms of sampling frequencies.}\looseness=-1

\section{Conclusions}\label{sec:conclusion}
We have presented a novel approach for safe \revision{event-triggered learning of GP models for CBF-based continuous-time controllers.} For achieving this, we design event triggers, which update the model using data generated online, such that they provide a sufficient excitation to efficiently reduce uncertainty. \revision{We show that finite control inputs are sufficient for maintaining feasibility of the CBF safety conditions with high probability and exclude} Zeno behavior with the proposed triggering scheme. \revision{Given a suitable GP prior and initial state, we prove that our approach probabilistically guarantees safety, which is illustrated} in numerical simulations of an adaptive cruise control system. Future work will focus on \revision{a complexity analysis and the extension of 
the proposed framework to scenarios with unknown CBFs. Moreover, it will address the problem of co-designing CBFs and the event-triggered learning scheme to ensure feasibility with given control input constraints.}\looseness=-1 
\section*{Appendix}
\label{appendix:proofs}

\subsection{Proof of \cref{lem:GPbound}}

\revision{By choosing $\tau_i = \nicefrac{\underline{\sigma}_i}{(L_i+L_{\mu_i}+L_{\sigma_i})}$ in \cite[Lemma 1]{Lederer2021HowTD} and slightly adapting this result, 
we obtain
    \begin{align*}
        \!|\mu_{i}(\bm{x})\!-\!c_i(\bm{x})| \! \leq \!\frac{\sqrt{\beta_i}}{2}\sigma_i(\bm{x}) \!+\! \frac{(L_i\!+\!L_{\mu_i} \!+\! \frac{\sqrt{\beta_i}}{2}L_{\sigma_i})\underline{\sigma}_i}{(L_i+L_{\mu_i}+L_{\sigma_i})}\!
    \end{align*}
    }%
    for all $\bm{x}\in\mathcal{X}$ with probability of at least $1-\nicefrac{\delta}{2(m+1)}$, dropping the arguments of $\beta(\cdot,\cdot)$ to simplify the presentation.  
    \revision{Note that Lipschitz constants of the mean and variance can be straightforwardly obtained due to stationarity and the sufficiently differentiable kernels $k_i(\cdot,\cdot)$ \cite{Lederer2021GaussianPb}. 
    Moreover, the sufficiently smooth kernel also ensures that sample functions $c_i(\cdot)$ of a GP admit Lipschitz constants $L_i$ defined in \cite[Theorem 3.2]{Lederer2019UniformEB} with probability $1-\frac{\delta}{2(m+1)}$ for each $i=0,\ldots,m$ individually.
    }%
    Since the right side of this inequality is linear in $\sigma_i(\cdot)$ and $\sigma_i (\boldsymbol{x}) \leq \max\{\sigma_{i}(\boldsymbol{x}),\underline{\sigma}_i\}$, we substitute $\tilde{\sigma}_i(\bm{x})=\max\{\sigma_{i}(\boldsymbol{x}),\underline{\sigma}_i\}$ to obtain 
    \begin{align*}
        |\mu_{i}(\bm{x})\!-\!c_i(\bm{x})| &\leq \frac{\sqrt{\beta_i}}{2}\tilde{\sigma}_i(\bm{x}) +\revision{ \frac{(L_i\!+\!L_{\mu_i} \!+\! \frac{\sqrt{\beta_i}}{2}L_{\sigma_i})\underline{\sigma}_i}{(L_i+L_{\mu_i}+L_{\sigma_i})}.}
    \end{align*}
    The last term in \eqref{eq:beta_i} is lower bounded by 0 as $\max_{\bm{x}\in\mathcal{X}}x_j-\min_{\bm{x}\in\mathcal{X}}x_j$ is lower bounded by $0$.
    The second term is lower bounded by $0$ since all constants are positive. As $\delta\in(0,1)$, it follows that $\frac{\beta_i}{4}\geq 2\log(\frac{m+1}{\delta})>1$, such that it holds that
    \begin{align*}
        (L_i\!+\!L_{\mu_i}\!+\!\frac{\sqrt{\beta_i}}{2}L_{\sigma_i})\underline{\sigma}_i\leq \frac{\sqrt{\beta_i}}{2}(L_i\!+\!L_{\mu_i}\!+\!L_{\sigma_i})\underline{\sigma}_i.
    \end{align*}
    Finally, noting that $\underline{\sigma}_i\leq \max\{\sigma_{i}(\boldsymbol{x}),\underline{\sigma}_i\}$, we obtain \eqref{eq:GPboundg}. The result is a consequence of the union bound \revision{over the probabilities of error bounds for all $i=0,\ldots,m$ and the probabilities of $L_i$ being Lipschitz constants.} 

\subsection{Proof of \cref{prop:safety}}

    \revision{
    Using the regression error bound in \cref{lem:GPbound}, we have
    $\hat{c}_0(\bm{x})\leq c_0(\bm{x})$
    with probability $1-\delta$. Moreover,
    \begin{align*}
        \|\bm{Q}(\bm{x})\bm{u}\| 
        &\geq \frac{\|\bm{Q}(\bm{x})\bm{u}\|_1}{\sqrt{m}} =
        \sum\limits_{i=1}^m |u_i| \sqrt{\beta_i(\delta, \underline{\sigma}_i)}\tilde{\sigma}_i(\bm{x}),
    \end{align*}
    such that the GP error bound in \cref{lem:GPbound} yields 
    $\hat{\bm{c}}^T(\bm{x})\bm{u}-\|\bm{Q}(\bm{x})\bm{u}\|\leq \bm{c}^T(\bm{x})\bm{u}$
    with probability $1-\delta$.
    Thus, satisfaction of \eqref{eq:cs_cons} implies satisfaction of \eqref{eq:cbf_const_base} with probability $1-\delta$. 
    Due to \eqref{eq:GP_cond}, there exists an $i=1,\ldots,m$ such that we can pick 
    \begin{align*}
        u_j=\begin{cases}
            0&\text{if } j\neq i\\
            -\frac{\hat{c}_0(\bm{x})}{\mu_i(\bm{x})-\mathrm{sgn}(\mu_i(\bm{x}))\sqrt{\beta_i(\delta, \underline{\sigma}_i)}\tilde{\sigma}_i(\bm{x})} &\text{if } i=j
        \end{cases}
    \end{align*}
    when $\hat{c}_0(\bm{x})<0$. If $\hat{c}_0(\bm{x})\geq 0$, the trivial choice $\bm{u}=\bm{0}$ is an admissible input in \eqref{eq:cs_cons}.
    Consequently, the optimization problem \eqref{eq:opt_QP} is feasible for all $\bm{x}\in\mathcal{C}$. Finally, a minor adaptation of \cite[Proposition 5.4]{Mestres2023FeasibilityAR} ensures that the control law is Lipschitz on $\mathcal{C}$ since $\hat{c}_0(\cdot)$, $\hat{\bm{c}}(\cdot)$ and $\bm{Q}(\cdot)$ are twice continuously differentiable due to the twice continuously differentiable kernels, $\|\bm{Q}(\bm{x})\bm{u}\|=0$ only if $\bm{u}=0$, and the constraint \eqref{eq:cs_cons} is inactive for $\hat{c}_0(\bm{x})=0$. 
    Hence, safety with probability $1-\delta$ follows from \cref{lem:cbf}.}

\subsection{Proof of \cref{prop:learn_policy}}

For $\nicefrac{-\hat{c}_0(\bm{x})}{|\xi_i(\bm{x})|}\geq \bar{u}_{\mathrm{GP}}$, \eqref{eq:cbf_const_base} is satisfied with probability $1-\delta$ following the same argumentation as in the proof of \cref{prop:safety} due to \eqref{eq:opt_trig}. \revision{If $\nicefrac{-\hat{c}_0(\bm{x})}{|\xi_i(\bm{x})|}< \bar{u}_{\mathrm{GP}}$, we have 
    \begin{align*}
        &\hat{\bm{c}}^T(\bm{x})\bar{\bm{\pi}}_i(\bm{x})-\|\bm{Q}(\bm{x})\bar{\bm{\pi}}_i(\bm{x})\| \\
        &\geq(|\mu_i(\bm{x})|-\sqrt{m\beta_i(\delta,\underline{\sigma}_i})\tilde{\sigma}_i(\bm{x}))\bar{u}_{\mathrm{GP}} \geq -\hat{c}_0(\bm{x})
    \end{align*}
    such that \eqref{eq:cbf_const_base} also holds with probability $1-\delta$, which concludes the first part of the proof.} For the second part, 
    we firstly consider the case that $\sigma_i^+(\bm{x})\leq \underline{\sigma}_i^+$, for which we have
    \begin{align}
    \frac{|\mu_i^+(\bm{x})|}{\sqrt{m\beta_i(\delta, \underline{\sigma}_i^{+})} \tilde{\sigma}_i^+(\bm{x})} 
    & \revision{\geq\frac{|\mu_i(\bm{x})|-\sqrt{\beta_i(\delta, \underline{\sigma}_i)} \tilde{\sigma}_i(\bm{x})}{\sqrt{m\beta_i(\delta, \underline{\sigma}_i^{+})} \underline{\sigma}_i^{+}}-\frac{1}{\sqrt{m}}} \notag \\
    &\revision{\geq \frac{\epsilon\sqrt{\beta_i(\delta, \underline{\sigma}_i)} \tilde{\sigma}_i(\bm{x})}{\sqrt{m\beta_i(\delta, \underline{\sigma}_i^{+})} \underline{\sigma}_i^{+}} -\frac{1}{\sqrt{m}}}\notag\\
    &\geq 1+\epsilon+\gamma\notag
    \end{align}
    with probability $1-\delta$. \revision{Here, the first line follows from applying \cref{lem:GPbound} twice, the second line is due to \eqref{eq:prior_safety_margin}, 
    and the last line is due to our choice of $\underline{\sigma}_i^{+}$ in \eqref{eq:sigma_under_cond}. 
    Observe that $\underline{\sigma}_i^{+}$ is well-defined since we can always find a sufficiently small value satisfying \eqref{eq:sigma_under_cond} due to the logarithmic dependency of $\beta_i(\cdot,\cdot)$ on $\underline{\sigma}_i^{+}$.} In the case that $\sigma_i^+(\bm{x})\geq \underline{\sigma}_i$, we can proceed analogously to obtain
    \begin{align*}
        \sigma_i^{+}(\bm{x})\leq \frac{\revision{\epsilon \sqrt{\beta_i(\delta,\underline{\sigma}_i)}\tilde{\sigma}_i(\bm{x})} }{(1+\epsilon+\gamma+\revision{\frac{1}{\sqrt{m}}})\sqrt{\revision{m}\beta_i(\delta,\underline{\sigma}_i^{+})}}
    \end{align*}
    as condition for ensuring the satisfaction of \eqref{eq:cond_event}.
    Due to \cite[Theorem 1]{Lederer2021HowTD}, this can be enforced through the requirement
    \begin{align*}
        s_i^2-\frac{s_i^4\bar{u}_{\mathrm{GP}}^2}{s_0^2+\bar{u}_{\mathrm{GP}}^2s_i^2+\sigma_{\mathrm{on}}^2}
        \leq \frac{\revision{\epsilon^2 \beta_i(\delta,\underline{\sigma}_i)\tilde{\sigma}_i^2(\bm{x})} }{(1+\epsilon+\gamma+\revision{\frac{1}{\sqrt{m}}})^2 \revision{m}\beta_i(\delta,\underline{\sigma}_i^+)}
    \end{align*}
    on the new data pair $([\bm{x}^T,\bar{\bm{\pi}}_i(\bm{x})]^T,y)$. Solving this inequality for \revision{$\bar{u}_{\mathrm{GP}}$ and ignoring negative terms under the square root} yields \eqref{eq:underline_u_GP}, such that the updated GP will satisfy \eqref{eq:cond_event}.

\subsection{Proof of \cref{th:safe_learn}}

Due to the GP model update event and the satisfaction of \eqref{eq:GP_cond} at $t=t_N$, \cref{prop:learn_policy} guarantees that \eqref{eq:GP_cond} holds for all $t\in[t_N,t_{N+1})$, where $t_{N+1}$ is the next trigger time instance defined in \eqref{eq:trigger}. Note that the switching between $\bar{\bm{\pi}}_i(\cdot)$ and $\bm{\pi}(\cdot)$ at $t=t_N$ renders the control law defined through \cref{alg:control_algorithm} discontinuous and a function of the time $t$. Nevertheless, an extended solution of \eqref{eq:si_sys} exists \cite[Chapter 2, Theorem 1.1]{Coddington1955TheoryOO} and is unique \cite[Chapter 2, Theorem 2.2]{Coddington1955TheoryOO} as 
the right-hand side of \eqref{eq:si_sys} under the control law defined by \cref{alg:control_algorithm} is a measurable function in time $t$, uniformly Lipschitz continuous in $\bm{x}$, and bounded on the compact set $\mathcal{C}$ due to the Lipschitz continuity of $f(\cdot)$, $g(\cdot)$, $\bm{\pi}(\cdot)$ and boundedness of $\bar{\bm{\pi}}_i(\cdot)$ for every $N\in\mathbb{N}$. Since \eqref{eq:cs_cons} is satisfied with probability \revision{$1-\frac{6\tilde{\delta}}{\pi^2N}$} for all $\bm{x}\in\mathcal{C}$ during the time interval $[t_N,t_{N+1})$ due to Propositions \ref{prop:safety} and \ref{prop:learn_policy}, we employ Nagumo's theorem \cite[Theorem 4.7]{Blanchini2008} to probabilistically guarantee forward invariance \revision{for every initial state $\bm{x}(t_N)\in\mathcal{C}$ of the time interval $[t_N,t_{N+1})$. Finally, we can chain these intervals together using the union bound as $\bm{x}(t_0)\in\mathcal{C}$ and \eqref{eq:prior_safety_margin} holds at $t=t_0$, such that safety is guaranteed for all $t\in[t_0,t_{\bar{N}})$ for every $\bar{N}\in\mathbb{N}$ with probability $1-\sum_{N=1}^{\bar{N}} \frac{6\tilde{\delta}}{\pi^2N^2}\geq 1-\tilde{\delta}$.}

\subsection{Proof of \cref{prop:zeno}}

The triggering function $\Gamma_N(t)$, $N\in\mathbb{N}$, is Lipschitz continuous because the standard deviation and mean functions are Lipschitz continuous, and $\tilde{\sigma}_i(x)$ is positive. Moreover, the trajectory $\bm{x}(t)$ is Lipschitz continuous with respect to time since it is the solution of a differential equation defined through bounded closed-loop dynamics, as discussed in the proof of \cref{th:safe_learn}. Let $L_{\Gamma_N}$ denote the Lipschitz constant of $\Gamma_N(\cdot)$ and let $t_N$, $t_{N+1}$ denote two consecutive triggering times. Then, we have
\begin{align*}
    \Gamma_N(t_{N+1}) & \geq \Gamma_N (t_N) - \left | \Gamma_N (t_{N+1}) - \Gamma_N (t_N)\right |\\
    &\geq 1+\epsilon + \gamma -  L_{\Gamma_N}(t_{N+1}-t_N)
\end{align*}
due to Lipschitz continuity of $\Gamma(\cdot)$ and the upper bound for $\Gamma(t_N)$ in \cref{prop:learn_policy} based on the GP model update at $t_N$. 
Because of the triggering condition \eqref{eq:cond_event}, we know that the next event occurs when $\Gamma_N(t_{N+1})=1+\epsilon$. This implies that $1+\epsilon \geq 1+\epsilon+\gamma -L_{\Gamma_N} (t_{N+1}-t_N)$
must hold. Therefore, we obtain $(t_{N+1}-t_N) \geq \Delta_N = \frac{\gamma}{L_{\Gamma_N}}$,
which concludes the proof.\looseness=-1

\section*{References}
\vspace*{-4ex}
\bibliographystyle{ieeetr}
\bibliography{main}
\end{document}